\documentclass[a4paper]{article}

\usepackage{a4wide,amsmath,amssymb,amsthm,hyperref,multicol,tikz-cd}

\newcommand{\cofission}[1]{{\cofissionop({#1})}}
\newcommand{\cofissionop}{{\tilde{\fis}}}
\newcommand{\cofusion}[1]{{\cofusionop({#1})}}
\newcommand{\cofusionop}{{\tilde{\fus}}}
\newcommand{\converse}{\smallsmile}
\newcommand{\convex}[1]{{#1}{\updownarrow}}
\newcommand{\dom}{\mathit{dom}}
\newcommand{\down}[1]{{#1}{\downarrow}}
\newcommand{\dual}[1]{{#1}^\mathsf{d}}
\newcommand{\eqem}{\mathrel{=_\updownarrow}}
\newcommand{\eqh}{\mathrel{=_\downarrow}}
\newcommand{\eqs}{\mathrel{=_\uparrow}}
\newcommand{\fis}{\delta_i}
\newcommand{\fus}{\delta_o}
\newcommand{\icpl}[1]{{\sim}{#1}}
\newcommand{\id}{\mathit{id}}
\newcommand{\Id}{\mathit{Id}}
\newcommand{\ii}{\Cap}
\newcommand{\iiatoms}{{\mathsf{A}_\ii}}
\newcommand{\iione}{{1_\ii}}
\newcommand{\iu}{\Cup}
\newcommand{\iU}{\raisebox{-0.5ex}{\Large$\iu$}}
\newcommand{\iuatoms}{{\mathsf{A}_\iu}}
\newcommand{\iuone}{{1_\iu}}
\newcommand{\kleisli}{\Pow}
\newcommand{\Mult}{M}
\newcommand{\Pow}{\mathcal{P}}
\newcommand{\Rel}{\mathbf{Rel}}
\newcommand{\rto}{\leftrightarrow}
\newcommand{\seq}{\ast}
\newcommand{\seqint}{\odot}
\newcommand{\Set}{\mathbf{Set}}
\newcommand{\subem}{\mathrel{\sqsubseteq_\updownarrow}}
\newcommand{\subh}{\mathrel{\sqsubseteq_\downarrow}}
\newcommand{\subs}{\mathrel{\sqsubseteq_\uparrow}}
\newcommand{\syq}[2]{{#1} \div {#2}}
\newcommand{\up}[1]{{#1}{\uparrow}}

\newtheorem{theorem}{Theorem}[section]
\newtheorem{proposition}[theorem]{Proposition}
\newtheorem{lemma}[theorem]{Lemma}
\newtheorem{corollary}[theorem]{Corollary}

\theoremstyle{definition}
\newtheorem{example}[theorem]{Example}
\newtheorem{remark}[theorem]{Remark}

\begin{document}

\title{Determinism of Multirelations}
\author{Hitoshi Furusawa, Walter Guttmann and Georg Struth}
\maketitle

\begin{abstract}
  Binary multirelations can model alternating nondeterminism, for instance, in games or nondeterministically evolving systems interacting with an environment.
  Such systems can show partial or total functional behaviour at both levels of alternation, so that nondeterministic behaviour may occur only at one level or both levels, or not at all.
  We study classes of inner and outer partial and total functional multirelations in a multirelational language based on relation algebra and power allegories.
  While it is known that general multirelations do not form a category, we show that the classes of deterministic multirelations mentioned form categories with respect to Peleg composition from concurrent dynamic logic, and sometimes quantaloids.
  Some of these are isomorphic to the category of binary relations.
  We also introduce determinisation maps that approximate multirelations either by binary relations or by deterministic multirelations.
  Such maps are useful for defining modal operators on multirelations.
\end{abstract}

\section{Introduction}
\label{section.introduction}

This is the second article in a trilogy on the inner structure of multirelations~\cite{FurusawaGuttmannStruth2023a}, the determinisation of such relations and their algebras of modal operators~\cite{FurusawaGuttmannStruth2023c}.

Multirelations are binary relations of type $X \rto \Pow Y$.
As explained in the first article of this trilogy, they are models of alternating angelic and demonic nondeterminism, while arbitrary relations are standard models of angelic nondeterminism without alternation.
Each element in $X$ can be related by such a multirelation, at the outer or angelic level of nondeterminism, to one subset or many subsets of $Y$, or to no set at all, and within each of these subsets, at the inner or demonic level of determinism, to one element or many elements, or to no element at all.
Multirelations have been used to describe the semantics of programs with both angelic and demonic nondeterminism \cite{Rewitzky2003}.

In the first part of this trilogy we have studied the inner or demonic structure of multirelations, which complements the usual angelic boolean structure on relations.
A typical inner operation is the inner union of two multirelations $R$, $S$ of the same type: if $R$ and $S$ relate an element $a$ with the sets $B$ and $C$, respectively, then $R \iu S$ relates $a$ with the set $B \cup C$.
Inner intersection and inner complementation can then be defined in the obvious way, performing set-intersection or set-complementation on the second components of pairs.

We have also discussed notions of inner univalence or inner partial functionality, inner totality and inner determinism or functionality, which complement the standard outer notions from relation algebra.
A multirelation is inner univalent if every element in its codomain is either empty or a singleton set.
This means that mapping to the empty set represents inner partiality.
Consequently, a multirelation is inner total if every element in its domain is related to a non-empty set, and it is inner deterministic if it is inner univalent and inner total.
By contrast, an outer univalent multirelation is (the graph of) a partial function, an outer total multirelation relates every element with some set (including the empty one) and an outer deterministic multirelation is (the graph of) a function.
Intuitively, inner univalent multirelations can thus be seen as angelic multirelations that do not allow any inner or demonic choices, while outer univalent multirelations can be seen as demonic, as they do not allow any outer or angelic choices~\cite{Rewitzky2003}.
Inner deterministic multirelations are therefore strictly angelic, as empty inner choices are not permitted, while outer deterministic multirelations are strictly demonic, as empty outer choices are impossible.

In this article we study the structure of inner and outer univalent and deterministic multirelations in an algebraic language~\cite{FurusawaKawaharaStruthTsumagari2017} that combines features of relation algebra~\cite{Schmidt2011}, quantaloids~\cite{Pitts1988,Rosenthal1996} and power allegories~\cite{FreydScedrov1990,BirdMoor1997} with specific operations for multirelations.
We also consider the determinisation of multirelations either by relations or by deterministic multirelations.

Apart from the operations on the inner and outer structure mentioned, we consider the Peleg composition of multirelations~\cite{Peleg1987}, which comes from concurrent dynamic logic and is one of several possible compositions for multirelations.
Multirelations under Peleg composition do not form categories because this operation is not associative.
Yet specific classes of multirelations do, for instance the classes of deterministic or univalent multirelations (Proposition~\ref{proposition.univalent-cat}).
In Propositions~\ref{prop.det-cat} and~\ref{proposition.peleg-rel-inner-univalent} we show that inner univalent and inner deterministic multirelations form categories with respect to Peleg composition as well.
In particular, the power transpose map from power allegories is in fact a functor from $\Rel$ to categories of inner deterministic multirelations with respect to Peleg composition.
Moreover, the categories of inner and outer deterministic multirelations are isomorphic to the category $\Rel$ of sets and relations and its enrichment in the form of quantaloids (Proposition~\ref{proposition.peleg-rel}).

We further introduce an operation that approximates multirelations by relations and their isomorphic inner and outer deterministic multirelations.
These determinisation maps on multirelations, which we call \emph{fusion} and \emph{fission} maps, are related to the original multirelation by Galois connections with respect to one of the inner preorders (Proposition~\ref{proposition.eta-alpha-lambda-galois}).
They are also functors between the categories of inner and outer deterministic multirelations (Corollary~\ref{cor.det-cat2}), and the inner and outer deterministic multirelations arise as their fixpoints.
The determinisation maps are further used in the proof that inner univalent multirelations form a category.

As in~\cite{FurusawaGuttmannStruth2023a}, we work in concrete extensions and enrichments of $\Rel$ throughout this article, but with a view towards future axiomatic approaches.
Once again we have used the Isabelle/HOL proof assistant to check many results in this article, and have developed a substantial library for reasoning with multirelations~\cite{GuttmannStruth2023}, and more generally with concrete power allegories.
Nevertheless we did not aim at a complete formalisation and our article is self-contained without the Isabelle libraries.


\section{Relations and Multirelations}
\label{section.relation-multirelation}

First we recall the basics of binary relations and multirelations.
See~\cite{FurusawaGuttmannStruth2023a} and the references therein for details.
Our algebraic language of concrete relations and multirelations is again based on enrichments of the category $\Rel$ of sets and relations.
Yet in contrast to~\cite{FurusawaGuttmannStruth2023a} we extend the standard calculus of relations~\cite{Schmidt2011} with concepts from power allegories~\cite{FreydScedrov1990} using in particular the connection with the monad of the powerset functor in $\Set$~\cite{BirdMoor1997}, and with multirelational concepts, as in~\cite{FurusawaKawaharaStruthTsumagari2017}.
The richness of this language sometimes prevents us from listing all properties used in calculations and proofs -- we often refer to ``standard'' properties instead.
We extend the dependency list of relational and multirelational concepts with respect to a small basis from~\cite{FurusawaGuttmannStruth2023a} to the additional concepts needed here in Appendix~\ref{section.basis}.

\subsection{Binary relations}
\label{subsection.binary-relations}

Following~\cite{FurusawaGuttmannStruth2023a}, we work in the category $\Rel$ with sets as objects, relations as arrows, relational composition as arrow composition and the identity or diagonal relations as identity arrows.
It forms a modular quantaloid~\cite{Rosenthal1996}, where a quantaloid is a category enriched in the symmetric closed monoidal category of sup-lattices.
For $\Rel$, the tensor yields relational composition and sup-preservation; by closure, relational composition has two residuals as right adjoints; as a sup-lattice, it has arbitrary sups and infs, that is, unions and intersections.
$\Rel$, in particular, is even a complete atomic boolean algebra.
As a modular quantaloid, $\Rel$ has the relational converse as an involution which satisfies the modular or Dedekind law of relation algebra.

We write $X \rto Y$ for the homset $\Rel(X,Y)$, $\Id_X$ for the identity relations on $X$, $\emptyset_{X,Y}$ for the least and $U_{X,Y}$ for the greatest element in $X \rto Y$, $-R$ for the complement of $R$ and $S-R$ for the relative complement $S \cap -R$, $RS$ for the relational composition of relations $R$, $S$ of suitable type, $R / S$ and $R \backslash S$ for the left and right residuals of $R$ and $S$ and $R^\converse$ for the converse of $R$.
The modular law is the property $R S \cap T \subseteq (R \cap T S^\converse) S$.

We need the properties $T \backslash S = (S^\converse / T^\converse)^\converse$, $T / S = -(-T S^\converse)$ and $T \backslash S = -(T^\converse (-S))$ of residuals.
We also need the following concepts:
\begin{itemize}
\item the \emph{symmetric quotient} $\syq{T}{S} : X \rto Y$ of $T : X \rto Z$ and $S : Y \rto Z$, defined as $\syq{T}{S} = (T \backslash S) \cap (T^\converse / S^\converse)$,
\item \emph{tests}, which are relations $R \subseteq \Id$, and whose relational composition is intersection,
\item the \emph{domain} map $R : X \rto Y, R \mapsto \Id_X \cap R R^\converse = \Id_X \cap R U_{Y,X} = \{ (a,a) \mid \exists b .\ R_{a,b} \}$.
\end{itemize}
Domain elements and tests form the same full subalgebra of $\Rel(X,X)$ for any $X$, a complete atomic boolean algebra.

As the title of this article indicates, we are particularly interested in the following properties.
The relation $R : X \rto Y$ is
\begin{itemize}
\item \emph{outer total} if $\dom(R) = \Id_X$, or equivalently $\Id_X \subseteq R R^\converse$,
\item \emph{outer univalent}, or a \emph{partial function}, if $R^\converse R \subseteq \Id_Y$,
\item \emph{outer deterministic}, or a \emph{function}, if it is total and univalent.
\end{itemize}
Functions as deterministic relations in $\Rel$ are of course graphs of functions in $\Set$.
They also model programs as a subset of nondeterministic specifications in program refinement calculi.
We need the relational law $P Q \cap S = (P \cap S Q^\converse) Q$ for outer univalent $Q$~\cite{SchmidtStroehlein1989} in calculations.

We further write $R|_A$ for the restriction of relation $R$ to domain elements in the set $A$, $R(A)$ for the relational image of $A$ under $R$ and $R(a)$ for $R(\{a\})$.

Next we recall the basic concepts from power allegories~\cite{FreydScedrov1990,BirdMoor1997}.
The isomorphism between relations in $X \rto Y$ and nondeterministic functions in $X \to \Pow Y$ in $\Set$ can be expressed in $\Rel$.
Nondeterministic functions $X \to \Pow Y$ in $\Set$ are of course functions $X \rto \Pow Y$ in $\Rel$.

The \emph{power transpose}
\begin{equation*}
  \Lambda : (X \rto Y) \to (X \rto \Pow Y), \, R \mapsto \left\{ (a,R(a)) \mid a \in X \right\}
\end{equation*}
maps relations $X \rto Y$ to functions in $X \rto \Pow Y$, which are graphs of the nondeterministic functions $X \to \Pow Y$ in $\Set$.
In the other direction, relational postcomposition with the has-element relation $\ni_Y : \Pow Y \rto Y$, the converse of the membership relation $\in_Y : Y \rto \Pow Y$, maps relations and therefore functions in $X \rto \Pow Y$ to relations in $X \rto Y$.
We henceforth write $\alpha = (-){\ni}$.
This function satisfies
\begin{equation*}
  \alpha : (X \rto \Pow Y) \to (X \rto Y), \, R \mapsto \left\{ (a,b) \mid b \in \bigcup R(a) \right\}.
\end{equation*}
Algebraically, $\Lambda(R) = \syq{R^\converse}{\in}$, and we will see below how $\ni$ and $\alpha$ can be expressed in terms of basic relational and multirelational operations.

\begin{lemma}
  \label{lemma.lambda-alpha-props}
  Let $R : X \to Y$ and let $f : X \to \Pow Y$ be deterministic.
  Then
  \begin{enumerate}
  \item $f = \Lambda(R) \Leftrightarrow R = \alpha(f)$,
  \item $\alpha(\Lambda(R)) = R$ and $\Lambda(\alpha(f)) = f$,
  \item $f \Lambda(R) = \Lambda(f R)$ and $\Lambda(\ni_X) = \Id_{\Pow X}$.
  \end{enumerate}
\end{lemma}

The following diagram therefore commutes for functions in $X \rto \Pow Y$:
\begin{equation*}
  \begin{tikzcd}
    X \rto \Pow Y \ar[r, "\id"] \ar[d, "\alpha"'] & X \rto \Pow Y \ar[d, "\alpha"] \\
    X \rto Y \ar[r, "\id"'] \ar[ur, "\Lambda"] & X \rto Y
  \end{tikzcd}
\end{equation*}
It follows that $\Lambda$ and $\alpha$ form a bijective pair.

We also need the \emph{relational image functor}
\begin{equation*}
  \Pow : (X \rto Y) \to (\Pow X \rto \Pow Y), \, R \mapsto \Lambda({\ni_X} R).
\end{equation*}
Expanding definitions, $\Pow(R) = \{ (A,R(A)) \mid A \subseteq X \}$, so that the relational image, given by the covariant powerset functor in $\Set$, is coded again as a graph.
It is deterministic by definition.
As a functor, it satisfies of course $\Pow(R S) = \Pow(R) \Pow(S)$ and $\Pow(\Id) = \Id$.

The unit and multiplication of the powerset monad are recovered relationally as $\eta_X : X \rto \Pow X$ and $\mu_X : \Pow^2 X \rto \Pow X$ such that $\eta_X = \Lambda (\Id_X)$ and $\mu_X = \Pow({\ni_X})$.
Alternatively, $\eta_X = \syq{\Id_X}{\in_X}$ and, expanding definitions, $\eta_X = \{ (a,\{a\}) \mid a \in X \}$.

\begin{lemma}
  \label{lemma.pow-props}
  Let $R : X \rto Y$, $S : Y \rto Z$, let $f : X \rto Y$ be deterministic.
  Then
  \begin{enumerate}
  \item $\Lambda(R S) = \Lambda(R) \Pow(S)$,
  \item $\eta \Pow(R) = \Lambda(R)$ and $\alpha(\eta \Pow(R)) = R$, hence $\Pow$ has a right inverse,
  \item $\Lambda(f) = f \eta$,
  \item $\eta$ and $\mu$ are natural transformations: $\eta \Pow(f) = f \eta$ and $\Pow^2(f) \mu = \mu \Pow (f)$,
  \item the monad axioms hold: $\Pow(\mu) \mu = \mu \mu$, $\Pow(\eta) \mu = \Id$ and $\eta \mu = \Id$,
  \item $\alpha(\eta) = \Id$.
  \end{enumerate}
\end{lemma}

Rather unsurprisingly, $\Pow$ does not form a monad on $\Rel$; it only does so on its wide subcategory $\Set$ (up to isomorphism).

The following relations are standard in relation algebra and can be defined in power allegories:
\begin{itemize}
\item the \emph{subset relation} $\Omega_Y = {\in_Y} \backslash {\in_Y} = \{ (A,B) \mid A \subseteq B \subseteq Y \}$,
\item the \emph{complementation relation} $C = \syq{\in_Y}{-\in_Y} = \{ (A,-A) \mid A \subseteq Y \}$.
\end{itemize}

We need the following technical lemma in proofs.
\begin{lemma}
  \label{lemma.rel-mod-props}
  Let $R : X \rto Y$.
  Then
  \begin{enumerate}
  \item $\Lambda(R) C = \Lambda(-R)$,
  \item $\Lambda(R) \Omega = R^\converse \backslash {\in} = ({\ni} / R)^\converse$.
  \end{enumerate}
\end{lemma}

\begin{proof}
  For (1), $\Lambda(R) C = \Lambda(R) \Lambda(-{\ni}) = \Lambda(\Lambda(R) (-{\ni})) = \Lambda(-(\Lambda(R) {\ni})) = \Lambda(-R)$.
  This uses properties of Lemma~\ref{lemma.lambda-alpha-props} and determinism of $\Lambda(R)$.

  For (2), $\Lambda(R) \Omega = \Lambda(R) (-({\ni} (-{\in}))) = -(\Lambda(R) {\ni} (-{\in})) = -(R (-{\in})) = R^\converse \backslash {\in}$, using the definition of $\Omega$ and determinism of $\Lambda(R)$ in the second step, Lemma~\ref{lemma.lambda-alpha-props} in the fourth and properties of residuals in the third.
\end{proof}

Finally, for $R, S : X \rto Y$, we write $S \subseteq_d R$ if $S$ is univalent, $\dom(S) = \dom(R)$ and $S \subseteq R$.
This allows us to decompose any relation as $R = \bigcup_{S \subseteq_d R} S$~\cite[Lemma 2.1]{FurusawaGuttmannStruth2023a}.

\subsection{Multirelations}
\label{SS:multirelations}

A \emph{multirelation} is an arrow $X \rto \Pow Y$ in $\Rel$ and therefore a doubly-nondeterministic function $X \to \Pow^2 Y$ in $\Set$.
We write $\Mult(X,Y)$ for the homset $X \rto \Pow Y$.
Multirelations do not form a category: the double powerset functor does not yield a suitable monad~\cite{KlinSalamanca2018} and hence no associative composition with suitable units~\cite{FurusawaKawaharaStruthTsumagari2017}.

The \emph{Peleg composition}~\cite{Peleg1987} $\seq : (X \rto \Pow Y) \times (Y \rto \Pow Z) \to (X \rto \Pow Z)$ can be defined in terms of the \emph{Peleg lifting} $(-)_\seq : (X \rto \Pow Y) \to (\Pow X \rto \Pow Y)$ of multirelations, which in turn can be defined in terms of the \emph{Kleisli lifting} $(-)_\kleisli : (X \rto \Pow Y) \to (\Pow X \rto \Pow Y)$~\cite{FurusawaKawaharaStruthTsumagari2017}:
\begin{equation*}
  R_\kleisli = \Pow(\alpha(R)), \qquad
  R_\seq = \dom(R)_\seq \bigcup_{S \subseteq_d R} S_\kleisli, \qquad
  R \seq S = R S_\seq.
\end{equation*}

Expanding definitions,
\begin{align*}
  R_\kleisli & = \left\{ (A,B) \mid B = \bigcup R(A) \right\}, \\
  R_\seq & = \left\{ (A,B) \mid \exists f : X \to \Pow Y .\ f|_A \subseteq R \wedge B = \bigcup f(A) \right\}, \\
  R \seq S & = \left\{ (a,C) \mid \exists B .\ R_{a,B} \wedge \exists f : Y \to \Pow Z .\ f|_B \subseteq S \wedge C = \bigcup f(B) \right\}.
\end{align*}

The Kleisli lifting is the multirelational analogue of the Kleisli lifting or Kleisli extension in the Kleisli category of the powerset monad.
Its standard definition translates to multirelations.

\begin{lemma}
  \label{lemma.pow-klift-def}
  Let $R : X \rto \Pow Y$.
  Then $R_\kleisli = \Pow(R) \mu$.
\end{lemma}

It can also be seen as the relational image of the relational approximation of any multirelation using $\alpha$.
By definition, Kleisli liftings of multirelations are functions in $\Rel$.

The units of Peleg composition are given by the multirelations $\eta_X$; because of this, we henceforth also write $1_X$ for $\eta_X$.

\begin{lemma}
  \label{lemma.klift-ext}
  Let $R : X \rto \Pow Y$, $S : Y \rto \Pow Z$ and let $f : X \rto \Pow Y$ be a function.
  Then the laws $(R S_\kleisli)_\kleisli = R_\kleisli S_\kleisli$, $\eta_\kleisli = \Id$ and $\eta f_\kleisli = f$ of extension systems hold.
\end{lemma}

The standard properties $\Pow(R) = (R \eta)_\kleisli$ and $\mu = \Id_\kleisli$, which recover the powerset monad from its extension system, still hold for any multirelation.
Once again, the extension system axioms work only for deterministic multirelations -- the standard arrows of the Kleisli category of the powerset functor.

The interaction of Peleg composition with the outer structure is weak; see~\cite{FurusawaStruth2015a} for examples.
In particular, it is not associative; only $(R \seq S) \seq T \subseteq R \seq (S \seq T)$ holds.
Hence $(R S_\seq)_\seq$ need not be equal to $R_\seq S_\seq$, and multirelations do not form a category under Peleg composition.
The composition becomes associative if the third factor is union-closed~\cite{FurusawaKawaharaStruthTsumagari2017}.
Peleg composition also preserves arbitrary unions in its first argument.

Algebraic descriptions of univalent and deterministic multirelations are simple.

\begin{lemma}[\cite{FurusawaKawaharaStruthTsumagari2017}]
  \label{lemma.det-circ}
  Let $R : X \rto \Pow Y$.
  Then
  \begin{enumerate}
  \item $R = \dom(R) 1_X R_\kleisli$ and $R_\seq = \dom(R)_\seq R_\kleisli$ if $R$ is univalent,
  \item $R = 1_X R_\kleisli$ and $R_\seq = R_\kleisli$ if $R$ is deterministic.
  \end{enumerate}
\end{lemma}

As Kleisli liftings of multirelations are functions, it follows from (1) that Peleg liftings of univalent multirelations are univalent.
Alternatively, $R : X \rto \Pow Y$ is univalent if and only if, for all $S : X \rto \Pow Y$, $\dom(R) = \dom(S)$ and $S \subseteq R$ imply $S = R$~\cite{FurusawaKawaharaStruthTsumagari2017}.
Thus
\begin{equation*}
    S_\seq
  = \dom(S)_\seq \bigcup_{T \subseteq_d S} T_\kleisli
  = \bigcup_{T \subseteq_d S} \dom(S)_\seq T_\kleisli
  = \bigcup_{T \subseteq_d S} \dom(T)_\seq T_\kleisli
  = \bigcup_{T \subseteq_d S} T_\seq,
\end{equation*}
as the $T$ are univalent, and therefore
\begin{equation*}
    R \seq S
  = R \bigcup_{T \subseteq_d S} T_\seq
  = R \, \dom(S)_\seq \bigcup_{T \subseteq_d S} T_\kleisli
  = R \, \dom(S)_\seq \bigcup_{T \subseteq_d S} \Pow(\alpha(T)).
\end{equation*}

Univalent multirelations have stronger algebraic properties.

\begin{lemma}[\cite{FurusawaKawaharaStruthTsumagari2017}]
  \label{lemma.univalent-assoc}
  Let $R$, $S$ and $f$ be composable multirelations and $f$ univalent.
  Then
  \begin{enumerate}
  \item the laws $(S f_\seq)_\seq = S_\seq f_\seq$, $\eta_\seq = \Id$ and $\eta R_\seq = R$ of extension systems hold,
  \item $(R \seq S) \seq f = R \seq (S \seq f)$.
  \end{enumerate}
\end{lemma}

\begin{proposition}
  \label{proposition.univalent-cat}
  The univalent and the deterministic multirelations form categories with sets as objects, multirelations as arrows, Peleg composition and $1_X$ as identity arrows.
\end{proposition}

\begin{proof}
  Peleg composition of univalent and therefore deterministic multirelations is associative by Lemma~\ref{lemma.univalent-assoc}.
  Multirelations $1_X$ are deterministic and hence univalent.
  It remains to show that $\seq$ preserves univalence and determinism.
  If $R$ and $S$ are composable univalent multirelations, then $R \seq S= R S_\seq$ is univalent because $S_\seq$ is univalent and relational composition preserves univalence.
  If $R$ and $S$ are also total, then $R \seq S = R S_\kleisli$ is total, because $S_\kleisli$ is deterministic and relational composition preserves totality.
\end{proof}

See Proposition~\ref{proposition.peleg-rel} for an alternative proof for deterministic multirelations and the end of Section~\ref{SS:det-cat} for a more thorough structural analysis of related properties, including the relationship of the Kleisli lifting of multirelations with the Kleisli category of the powerset functor.

Definitions of inner univalence, inner totality and inner determinism depend on inner operations on multirelations, which have been studied in detail in~\cite{FurusawaGuttmannStruth2023a}, based on previous work in~\cite{Rewitzky2003,FurusawaStruth2015a,FurusawaStruth2016}.
Algebraic definitions relative to a small basis can be found in Appendix~\ref{section.basis}.

For $R, S : X \rto \Pow Y$, one can define
\begin{itemize}
\item \emph{inner union} $R \iu S = \{ (a,A \cup B) \mid R_{a,A} \wedge S_{a,B} \}$ with unit $\iuone = \{ (a,\emptyset) \mid a \in X \}$,
\item \emph{inner complementation} $\icpl{R} = RC = \{ (a,-A) \mid R_{a,A} \}$,
\item the set of \emph{atoms} $\iuatoms = \{ (a,\{b\}) \mid a \in X \wedge b \in Y \}$ in $\Mult(X,Y)$.
\end{itemize}
The \emph{inner intersection} and its unit are then obtained as $R \ii S = \icpl(\icpl{R} \iu \icpl{S})$ and $\iione = \icpl{\iuone}$.
The operations $\iu$ and $\ii$ are associative and commutative, but need not be idempotent.

The multirelation $R:X\rto Y$ is
\begin{itemize}
  \item \emph{inner total} if $R \subseteq -\iuone$, that is, $B$ is non-empty for each $(a,B) \in R$,
  \item \emph{inner univalent} if $R \subseteq \iuatoms \cup \iuone$, that is, $B$ is either a singleton or empty for each $(a,B) \in R$,
  \item \emph{inner deterministic} if it is inner total and inner univalent, in which case $B \subseteq Y$ is a singleton set whenever $R_{a,B}$ for some $a \in X$.
\end{itemize}

Sets of inner univalent, total and deterministic multirelations can be characterised as fixpoints.

\begin{lemma}[{\cite[Lemma 3.9]{FurusawaGuttmannStruth2023a}}]
  \label{lemma.det-fix-ref}
  ~
  \begin{enumerate}
  \item The inner univalent multirelations are the fixpoints of $(-) \cap (\iuatoms \cup \iuone)$.
  \item The inner total multirelations are the fixpoints of $(-) - \iuone$.
  \item The inner deterministic multirelations are the fixpoints of $(-) \cap \iuatoms$ and $(-) 1^\converse 1$.
  \end{enumerate}
\end{lemma}

We also need the following closures and inner preorder, which compare the inner nondeterminism of multirelations.
For $R : X \rto \Pow Y$,
\begin{itemize}
\item the \emph{up-closure} $\up{R} = R \iu U = R \Omega = \{ (a,A) \mid \exists (a,B) \in R .\ B \subseteq A \}$ and the \emph{Smyth preorder} $R \subs S \Leftrightarrow S \subseteq \up{R}$ with equivalence $\eqs$,
\item the \emph{down-closure} $\down{R} = R \ii U = R \Omega^\converse = \{ (a,A) \mid \exists (a,B) \in R .\ A \subseteq B \}$ and the \emph{Hoare preorder} $R \subh S \Leftrightarrow R \subseteq \down{S}$ with equivalence $\eqh$.
\end{itemize}
The \emph{convex closure} can then be defined as $\convex{R} = \up{R} \cap \down{R}$, and the \emph{Egli-Milner preorder} as $R \subem S \Leftrightarrow R \subh S \wedge R \subs S$ with equivalence $\eqem$.
The up-closure and down-closure are related by inner duality.
Using up-closure, ${\in} = \up{1}$.
Moreover, $(\down{1})_\seq = (1 \cup \iuone)_\seq =\Omega^\converse$, and thus $\down{R} = R \seq \down{1}$ for all $R : X \rto \Pow Y$.
See Appendix~\ref{section.basis} and~\cite{FurusawaGuttmannStruth2023a} for context.


\section{Deterministic Multirelations}
\label{section.determinism}

Relations $X \rto Y$ embed into multirelations $X \rto \Pow Y$ in two natural ways: postcomposition with $1_Y$ lifts the elements in $Y$ to singleton sets in $\Pow Y$; taking the power transpose $\Lambda$ represents the standard equivalent nondeterministic function as a multirelation $X \rto \Pow X$.
The first embedding yields an inner deterministic multirelation, the second an outer deterministic one:
\begin{equation*}
  \begin{tikzcd}[sep=large]
    \{ (a,R(a)) \mid a \in X \} & R \ar[r, mapsto, "(-) 1_Y"] \ar[l, mapsto, "\Lambda"'] & \{ (a,\{b\}) \mid R_{a,b} \} .
  \end{tikzcd}
\end{equation*}
These embeddings extend to isomorphisms between the categories $\Rel$, categories of inner deterministic multirelations with Kleisli composition as arrow compositions and identity arrows $1_X$ and categories of outer deterministic multirelations with the same composition and identity arrows.
As $1_X$ is the unit of the powerset monad in relational form, we henceforth also write $\eta_X = (-) 1_X$.

The functions $\Lambda$, $\alpha$ and $\eta$ become functors in this setting.
The isomorphism between $\Rel$ and the category of outer deterministic multirelations is just that between $\Rel$ and the Kleisli category of the powerset monad in relational form.
That between $\Rel$ and the category of inner deterministic multirelations is trivial.
A formal proof that inner and outer deterministic multirelations form categories, however, requires some work.

Recall that a \emph{quantaloid} is a category in which each homset forms a complete lattice and where arrow composition preserves arbitrary sups in both arguments.

\subsection{Bijections between relations and deterministic multirelations}
\label{SS:det-bij}

First we study the bijections between relations and outer and inner deterministic multirelations in detail.
The results for outer deterministic relations are known.
Those for inner deterministic relations are new, but rather obvious.

\begin{lemma}
  \label{lemma.det-char}
  For every $R : X \rto Y$, $\Lambda(R)$ is outer deterministic and $\eta(R)$ inner deterministic.
\end{lemma}

\begin{proof}
  This is well known for outer determinism~\cite{BirdMoor1997}.
  For inner determinism,
  \begin{equation*}
    \eta(R) = R 1 \subseteq U 1 = \iuatoms
  \end{equation*}
  since composition preserves the order.
\end{proof}

Recall from Lemma~\ref{lemma.lambda-alpha-props} that $\alpha \circ \Lambda = \id_{X \rto Y}$, while $\Lambda \circ \alpha = \id_{X \rto \Pow Y}$ on outer deterministic multirelations.
Hence $\Lambda$ and $\alpha$ form a bijective pair between relations and outer deterministic multirelations.
A similar fact holds for $\eta$ and $\alpha$.
Before proving it, we mention a technical lemma.

\begin{lemma}
  \label{lemma.inner-det-props}
  Let $R$, $S$, $T$ be multirelations of appropriate types and $R$ inner deterministic.
  Then
  \begin{enumerate}
  \item $R \seq S = R 1^\converse S$,
  \item $R \seq (S \seq T)= (R \seq S) \seq T$,
  \item $\alpha(R) = R 1^\converse$.
  \end{enumerate}
\end{lemma}

\begin{proof}
  For (1) and (2), see~\cite{FurusawaGuttmannStruth2023a}.
  For (3), $\alpha(R) = R 1^\converse 1 {\ni} = R 1^\converse \alpha(1) = R 1^\converse$ using Lemma~\ref{lemma.pow-props}.
\end{proof}

\begin{lemma}
  \label{lemma.bij-alpha-one}
  The functions $\alpha$ and $\eta$ form a bijective pair between relations and inner deterministic multirelations.
\end{lemma}

\begin{proof}
  We need to check $\alpha \circ \eta_Y = \id_{X \rto Y}$ on relations and $\eta_Y \circ \alpha = \id_{X \rto \Pow Y}$ on inner deterministic multirelations.
  For the first identity, $\alpha(\eta(R)) = \alpha(R 1) = R \alpha(1) = R \, \Id = R$.
  For the second one, $\eta(\alpha(R)) = R 1^\converse 1 = R$, using Lemma~\ref{lemma.det-fix-ref}(3) and Lemma~\ref{lemma.inner-det-props}(3).
\end{proof}

The commutative diagram from Section~\ref{subsection.binary-relations} can thus be expanded:

\begin{equation*}
  \begin{tikzcd}
    X \rto \Pow Y \ar[r, "\id"] \ar[d, "\alpha"'] & X \rto \Pow Y \ar[d, "\alpha"] \\
    X \rto Y \ar[r, "\id"] \ar[ur, "\Lambda"] \ar[dr, "\eta"] & X \rto Y \\
    X \rto \Pow Y \ar[r, "\id"'] \ar[u, "\alpha"] & X \rto \Pow Y \ar[u, "\alpha"']
  \end{tikzcd}
\end{equation*}
The multirelations in the upper and lower rows are outer and inner deterministic, respectively.

Lemma~\ref{lemma.bij-alpha-one} implies that $S = \eta(R) \Leftrightarrow R = \alpha(S)$ holds for any $R : X \rto Y$ and inner deterministic $S : X \rto \Pow Y$.

\subsection{Categories of deterministic multirelations}
\label{SS:det-cat}

The bijections between relations and deterministic multirelations extend to isomorphisms between categories and quantaloids.
The maps $\Lambda$, $\alpha$ and $\eta$ become functors, and in fact isomorphisms.

\begin{lemma}
  \label{lemma.det-prop}
  Let $R : X \rto Y$ and $S : Y \rto Z$.
  Then
  \begin{enumerate}
  \item $\Lambda(R S) = \Lambda(R) \seq \Lambda(S)$ and $\Lambda(\Id_X) = 1_X$,
  \item $\eta(R S) = \eta(R) \seq \eta(S)$ and $\eta(\Id_X) = 1_X$,
  \item $\alpha (R \seq S) = \alpha(R) \alpha(S)$ and $\alpha(1_X) = \Id_X$ if $R$ and $S$ are inner or outer deterministic.
  \end{enumerate}
\end{lemma}

\begin{proof}
  For (1), $\Lambda(R S) = \Lambda(R) \Pow(S) = \Lambda(R) \Lambda(S)_\kleisli = \Lambda(R) \Lambda(S)_\seq = \Lambda(R) \seq \Lambda(S)$, where the second step holds by $\Lambda(S)_\kleisli = \Pow(\Lambda(S)) \mu = \Pow(\Lambda(S)) \Pow({\ni}) = \Pow(\Lambda(S){\ni}) = \Pow(S)$.

  For (2), $\eta(R) \seq \eta(S) = R 1 \seq S 1 = R 1 (S 1)_\seq = R (1 \seq S 1) = R S 1 = \eta(R S)$.

  For (3), suppose $R$ and $S$ are outer deterministic.
  Then $R \seq S$ is outer deterministic by Proposition~\ref{proposition.univalent-cat}.
  Hence $\alpha(R) \alpha(S) = \alpha(R \seq S)$ if and only if $\Lambda(\alpha(R) \alpha(S)) = R \seq S$ because $\Lambda$ and $\alpha$ form a bijective pair.
  Using this property again with (1),
  \begin{equation*}
    \Lambda(\alpha(R) \alpha(S)) = \Lambda(\alpha(R)) \seq \Lambda(\alpha(R)) = R \seq S.
  \end{equation*}
  The proof for inner determinism, $\alpha$ and $\eta$ is similar.
  In particular, if $R$ and $S$ are inner deterministic, then so is $R \seq S$, because $R \seq S = R \seq S 1^\converse 1 = R 1^\converse S 1^\converse 1 = (R \seq S) 1^\converse 1$, using Lemmas~\ref{lemma.det-fix-ref} and~\ref{lemma.inner-det-props}.
\end{proof}

\begin{proposition}
  \label{prop.det-cat}
  The inner and the outer deterministic multirelations form categories with respect to Peleg composition and the $1_X$.
  Both are isomorphic to $\Rel$.
\end{proposition}

\begin{proof}
  The bijections between relations, inner deterministic multirelations and outer deterministic multirelations together with Lemma~\ref{lemma.det-prop} imply that inner and outer determinism are closed under Peleg composition and that the $1_X$ are inner and outer deterministic.

  For associativity of Peleg composition, suppose $R$, $S$, $T$ are outer deterministic and composable.
  Then $R = \Lambda(\alpha(R))$ and likewise for $S$ and $T$.
  Hence
  \begin{align*}
      (R \seq S) \seq T
    & = (\Lambda(\alpha(R)) \seq \Lambda(\alpha(R))) \seq \Lambda(\alpha(T)) \\
    & = \Lambda (\alpha(R) \alpha(S) \alpha(T)) \\
    & = \Lambda(\alpha(R)) \seq (\Lambda(\alpha(R)) \seq \Lambda(\alpha(T))) \\
    & = R \seq (S \seq T),
  \end{align*}
  using Lemma~\ref{lemma.det-prop}.
  The proof of associativity for inner deterministic multirelations is similar.
  Hence both sets form categories with identity arrows $1_X$.

  For the isomorphisms, recall that $\Lambda$, $\alpha$ and $\eta$ form bijective pairs that preserve compositions by Lemma~\ref{lemma.det-prop}.
  So $\Lambda$ is a fully faithful functor from $\Rel$ to the category of outer deterministic multirelations, $\eta$ a fully faithful functor from $\Rel$ to the category of inner deterministic multirelations and $\alpha$ a fully faithful functor in the other directions.
  (The object components of these functors are identities, viewing multirelations $X \rto \Pow Y$ in $\Rel$ as arrows from $X$ to $Y$.)
\end{proof}

Peleg composition is thus a faithful representation of Kleisli composition of nondeterministic functions modelled as outer deterministic multirelations, and of relational composition of relations modelled as inner deterministic multirelations.

\begin{lemma}
  \label{lemma.lambda-em-isotone}
  ~
  \begin{enumerate}
  \item For relations $R$ and $S$, $R \subseteq S$ implies $\eta(R) \subseteq \eta(S)$ and $\Lambda(R) \subem \Lambda(S)$.
  \item For multirelations $R$ and $S$, $R \subseteq S$ implies $\alpha(R) \subseteq \alpha(S)$.
  \end{enumerate}
\end{lemma}

\begin{proof}
  We first show $\Lambda(R) \subs \Lambda(S)$, that is, $\Lambda(S) \subseteq \up{\Lambda(R)} = \Lambda(R) \Omega = R^\converse \backslash {\in}$ (see Lemma~\ref{lemma.rel-mod-props}(2)).
  By residuation this is equivalent to $R^\converse \Lambda(S) \subseteq {\in}$.
  Since $\Lambda(S)$ is a function, this is equivalent to $R \subseteq \Lambda(S) {\ni} = S$ using Lemma~\ref{lemma.lambda-alpha-props}, which is the assumption.
  Since $\Lambda(R)$ and $\Lambda(S)$ are functions, $\Lambda(R) \subem \Lambda(S)$ follows because $\subs$, $\subh$ and $\subem$ coincide on outer deterministic multirelations~\cite[Proposition 5.8]{FurusawaGuttmannStruth2023a}.
  The remaining claims follow since relational composition preserves $\subseteq$.
\end{proof}

\begin{remark}
  Consider relations $R = \emptyset$ and $S = \{ (a,a) \}$ on the set $\{a\}$.
  Then $R \subseteq S$, but $\Lambda(R) = \{ (a,\emptyset) \} \nsubseteq \{ (a,\{a\}) \} = \Lambda(S)$.
\end{remark}

The categories in Proposition~\ref{prop.det-cat} are enriched.
We define
\begin{equation*}
  \underset{i \in I}{\iU} R_i = \left\{ \left( a, \bigcup_{i \in I} A_i \right) \middle\vert \ \forall i \in I .\ (a,A_i) \in R_i \right\}
\end{equation*}
to capture one of the quantaloid structures that arise.
Before that, note that multirelations under Peleg composition and the outer operations do not form quantaloids: Peleg composition is not associative and does not preserve the sups needed~\cite{FurusawaStruth2015a,FurusawaStruth2016}.

\begin{proposition}
  \label{proposition.peleg-rel}
  The inner deterministic multirelations with $\bigcup$ and the outer deterministic multirelations with $\iU$ form quantaloids isomorphic to the quantaloid of binary relations.
\end{proposition}

\begin{proof}
  For the quantaloid of inner deterministic multirelations, recall that relational composition preserves arbitrary unions, hence so do the isomorphisms $\eta$ and $\alpha$ between $\Rel$ and the inner deterministic multirelations:
  \begin{equation*}
    \eta \left( \bigcup_{i \in I} R_i \right) = \bigcup_{i \in R} \eta(R_i) \qquad
    \text{ and } \qquad
    \alpha \left( \bigcup_{i \in I} S_i \right) = \bigcup_{i \in I} \alpha(S_i),
  \end{equation*}
  if all $S_i$ are inner deterministic.
  Inner determinism is therefore preserved by arbitrary unions and Peleg composition distributes over arbitrary unions of inner deterministic multirelations.

  For the quantaloid of outer deterministic multirelations, the isomorphisms $\alpha$ and $\Lambda$ between outer deterministic multirelations and $\Rel$ satisfy
  \begin{equation*}
    \Lambda \left( \bigcup_{i \in I} R_i \right) = \underset{i \in I}{\iU} \Lambda(S_i) \qquad
    \text{ and } \qquad
    \alpha \left( \underset{i \in I}{\iU} S_i \right) = \bigcup_{i \in I} \alpha(S_i)
  \end{equation*}
  if all $S_i$ are deterministic.
  Moreover, the definition of arbitrary inner unions implies that they preserve determinism.
  Hence Peleg composition distributes over arbitrary inner unions of deterministic multirelations.
  Inner union is idempotent on univalent and hence on deterministic multirelations~\cite[Lemma 3.6]{FurusawaGuttmannStruth2023a}.
  The order of the complete lattice can be defined via $R \leq S \Leftrightarrow R \iu S = S$, the natural order for deterministic multirelations~\cite[Lemma 5.9]{FurusawaGuttmannStruth2023a}.
\end{proof}

\begin{remark}
  \label{remark.kleisli}
  A Kleisli composition of multirelations can be defined as $R \circ_\kleisli S = R S_\kleisli$~\cite{FurusawaKawaharaStruthTsumagari2017}.
  It satisfies the standard identity $R \circ_\kleisli S = R \Pow (S) \mu$, is associative on arbitrary multirelations of appropriate type and has $1$ as its right unit, and as a left unit on outer deterministic multirelations.
  By Lemma~\ref{lemma.det-circ}, Peleg and Kleisli lifting coincide on the category of outer deterministic multirelations.
  Finally, the category of outer deterministic multirelations is isomorphic to the Kleisli category of the powerset functor in $\Set$, using the graph functor to map from $X \to \Pow Y$ to $X \rto \Pow Y$, which is clearly bijective.
\end{remark}

\subsection{Determinisation of multirelations}
\label{subsection.fusion}

The maps $\Lambda \circ \alpha$ and $\eta \circ \alpha$ approximate multirelations by relations modelled as isomorphic inner or outer deterministic multirelations.
They also form the isomorphism between the categories of inner and outer deterministic multirelations.

Let $R : X \rto \Pow Y$ be a multirelation.
The \emph{outer determinisation} or \emph{fusion} operation
\begin{equation*}
  \fus = \Lambda \circ \alpha
\end{equation*}
sends $R$ to the outer deterministic multirelation isomorphic to relation $\alpha(R)$.
The \emph{inner determinisation} or \emph{fission} operation
\begin{equation*}
  \fis = \eta \circ \alpha
\end{equation*}
sends $R$ to the inner deterministic multirelation isomorphic to $\alpha(R)$.
This is expressed in the following commuting diagram.
\begin{equation*}
  \begin{tikzcd}
    X \rto \Pow Y \ar[r, "\fus"] \ar[d, "\alpha"'] & X \rto \Pow Y \ar[d, "\alpha"] \\
    X \rto Y \ar[r, "\id"] \ar[ur, "\Lambda"] \ar[dr, "\eta"] & X \rto Y \\
    X \rto \Pow Y \ar[r, "\fis"'] \ar[u, "\alpha"] & X \rto \Pow Y \ar[u, "\alpha"']
  \end{tikzcd}
\end{equation*}

Set-theoretically,
\begin{equation*}
  \fus(R) = \{ (a,B) \mid B = \bigcup R(a) \} \qquad
  \text{ and } \qquad
  \fis(R) = \{ (a,\{b\}) \mid b \in \bigcup R(a) \}.
\end{equation*}

Composing the bijections in this diagram from bottom to top and vice versa yields the following corollary to Propositions~\ref{prop.det-cat} and~\ref{proposition.peleg-rel}.

\begin{corollary}
  \label{cor.det-cat2}
  The functors $\fis$ and $\fus$ are isomorphisms between the categories of inner deterministic and outer deterministic multirelations.
  They preserve the quantaloid structure with $\bigcup$ for inner deterministic multirelations and $\iU$ for outer deterministic ones.
\end{corollary}

For outer deterministic multirelations, therefore, $\fus \circ \fis = \id_{X \rto \Pow Y}$ and for inner deterministic ones, $\fis \circ \fus = \id_{X \rto \Pow Y}$ and we get the universal property $R = \fis(S) \Leftrightarrow S = \fus(R)$ for inner deterministic $R$ and outer deterministic $S$.
By functoriality, $\fis(R \seq S) = \fis(R) \seq \fis(S)$ if $R$, $S$ are outer deterministic and $\fus(R \seq S) = \fus(R) \seq \fus(S)$ if $R$, $S$ are inner deterministic.

\begin{corollary}
  \label{cor.det-fix2}
  The inner and outer deterministic multirelations are precisely the fixpoints of $\fis$ and $\fus$, respectively.
\end{corollary}

\begin{proof}
  If $R$ is inner deterministic, then $\fis(R) = \eta(\alpha(R)) = R$ by Lemma~\ref{lemma.bij-alpha-one}.
  If $\fis(R) = R$, then $R$ is inner deterministic by Lemma~\ref{lemma.det-char}.
  The proof for outer determinism is similar.
\end{proof}

The universal properties for $\alpha$ and $\Lambda$ or $\eta$ for relations and outer or inner deterministic multirelations generalise to Galois connections on arbitrary multirelations.
These use $\subseteq$ on relations and $\subh$ on multirelations.

\begin{proposition}
  \label{proposition.eta-alpha-lambda-galois}
  Let $R, S : X \rto \Pow Y$ and $T, V : X \rto Y$.
  Then
  \begin{enumerate}
  \item $\alpha(R) \subseteq T \Leftrightarrow R \subh \Lambda(T)$, $\eta(T) \subh S \Leftrightarrow T \subseteq \alpha(S)$ and $\fis(R) \subh S \Leftrightarrow R \subh \fus(S)$,
  \item $\alpha(R \ii S) = \alpha(R) \cap \alpha(S)$,
  \item $T \subseteq V \Rightarrow \Lambda(T) \subh \Lambda(V)$, $T \subseteq V \Rightarrow \eta(T) \subh \eta(V)$ and $R \subh S \Rightarrow \alpha(R) \subseteq \alpha(S)$,
  \item $\fus$ is a closure and $\fis$ an interior operator,
  \item $(\alpha,\Lambda)$ and $(\alpha,\eta)$ are epi-mono-factorisations of $\fus$ and $\fis$, both unique up to isomorphism,
  \item $\fus(R)$ is the $\subh$-least outer deterministic multirelation above $R$ and $\fis(R)$ the $\subh$-greatest inner deterministic multirelation below $R$.
  \end{enumerate}
\end{proposition}

\begin{proof}
  For (1), recall that $R \subh S \Leftrightarrow R \subseteq S \Omega^\converse$ and $ \Omega^\converse = {\ni} / {\ni}$.
  For the first Galois connection, $\alpha(R) \subseteq T \Leftrightarrow R \subseteq T / {\ni}$ using the standard Galois connection for left residuals.
  The claim then follows from $T / {\ni}= (\Lambda(T) {\ni}) / {\ni} = \Lambda(T) \Omega^\converse$, using a general law of residuals ($(R S) / Q = R (S / Q)$ for all composable relations $R$ and $S$ such that $R$ is deterministic) in the last step.
  For the second Galois connection, first suppose $\eta(T) \subseteq S \Omega^\converse$.
  Then $T \subseteq S \Omega^\converse {\ni} = \alpha(S)$, using $\alpha \circ \eta = \id$ and $\subseteq$-preservation of $\alpha$ in the first step, and $(R/ R) R = R$, which holds for all relations $R$, in the second one.
  Conversely, suppose $T \subseteq \alpha(S)$.
  Then $\eta(T) \subseteq S {\ni} 1 \subseteq S \Omega^\converse$, because ${\ni} 1 {\ni} = {\ni} \alpha(\Lambda(\Id)) = {\ni}$ and therefore ${\ni} 1 \subseteq \Omega^\converse$ by the Galois connection for left residuals.
  The third Galois connection is then immediate.

  Item (2) follows from a simple set-theoretic calculation.

  The first two properties in (3) follow directly from the Galois connections in (1).
  For the third one, $R \subh S$ implies $\alpha(R) \subseteq \alpha(S \Omega^\converse) \subseteq \alpha(S)$, where the first step uses $\subseteq$-preservation of $\alpha$ and the second $\Omega^\converse {\ni} = {\ni}$, like in (1).

  For (4), $\subh$-preservation of $\fus$ and $\fis$ follows from (3); $R \subh \fus(R)$ and $\fis(R) \subh R$ follow from (1).
  Further, $\fus \circ \fus = \fus$ and $\fis \circ \fis = \fis$ hold because $\alpha \circ \Lambda = \id = \alpha \circ \eta$.

  For (5), surjectivity of $\alpha$ and injectivity of $\Lambda$ and $\eta$ is immediate from the cancellation properties of the Galois connections in (1).
  For uniqueness, note that every function in $\Set$ has this property, and the proof is standard.

  For (6), $R \subh \fus(R)$ by (4).
  Now suppose $\fus(S) = S$ and $R \subh S$.
  Then $\fus(R) \subh \fus(S) = S$ by order-preservation of $\fus$.
  The proof for $\fis$ is similar.
\end{proof}

The properties in (1) and (2) can be summarised in the language of topos theory by saying that the adjunction $(\alpha,\Lambda)$ is an essential geometric morphism -- but of course $\Rel$ does not form a topos.
The properties in (3) and (4) are standard for Galois connections, however these are usually presented for two partial orders or for two preorders whereas we have a mixed case.
This is why we list and prove these properties.
Item (4) shows that $\fus$ is a monad and $\fis$ a comonad on $\subh$, both of which are idempotent.
As usual, the unit and counit are arrows $R \subh \fus(R)$ and $\fis(R) \subh R$, the multiplication and comultiplication are arrows $\fus(\fus(R)) \subh \fus(R)$ and $\fis(R) \subh \fis(\fis(R))$.

\begin{example}
 The following examples rule out Galois connections for set inclusion only.
  \begin{itemize}
  \item For $(\alpha,\Lambda)$, consider $R = \{ (a,\emptyset) \}$ in $X \rto \Pow Y$ and $S = \{ (a,b) \}$ in $X \rto Y$ for $X = \{ a \}$ and $Y = \{ b \}$.
        Then $\alpha(R) = \emptyset \subseteq S$, but $R \nsubseteq \{ (a,\{b\}) \} = \Lambda(S)$.
  \item For $(\Lambda,\alpha)$, consider $R = \emptyset$ in $X \rto Y$ and $S = \{ (a,\{b\}) \}$ in $X \rto \Pow Y$ for $X = \{a\}$ and $Y = \{b\}$.
        Then $R \subseteq \{ (a,b) \} = \alpha(S)$, but $\Lambda(R) = \{ (a,\emptyset) \} \nsubseteq S$.
  \item For $(\alpha,\eta)$, consider $R = \{ (a,\{b_1,b_2\}) \}$ in $X \rto \Pow Y$ and $S = \{ (a,b_1), (a,b_2) \}$ in $X \rto Y$ for $X = \{a\}$ and $Y = \{b_1,b_2\}$.
        Then $\alpha(R) = S$, but $R \nsubseteq \{ (a,\{b_1\}), (a,\{b_2\}) \} = \eta(S)$.
  \item For $(\eta,\alpha)$, consider $R = \{ (a,b_1), (a,b_2) \}$ in $X \rto Y$ and $S = \{ (a,\{b_1,b_2\}) \}$ in $X \rto \Pow Y$ for $X = \{a\}$ and $Y = \{b_1,b_2\}$.
        Then $R = \alpha(S)$, but $\eta(R) = \{ (a,\{b_1\}), (a,\{b_2\}) \} \nsubseteq S$.
  \end{itemize}
\end{example}

Fusion and fission can also be defined in terms of multirelations.
This requires two additional concepts from the inner structure: the set of \emph{co-atoms} $\iiatoms = \icpl{\iuatoms} = \{ (a,X - \{b\}) \mid a \in X \wedge b \in Y\}$ in $\Mult(X,Y)$ and the \emph{duality} operation $\dual{R} = -\icpl{R}$, which relates the inner and the outer structure.
See~\cite{FurusawaGuttmannStruth2023a} for details.
\begin{lemma}
  \label{lemma.fus-fis-explicit}
  Let $R : X \rto \Pow Y$.
  Then
  \begin{enumerate}
  \item $\fis(R) = \down{R} \cap \iuatoms$,
  \item $\down{\fus(R)} = -(\up{(-\fis(R) \cap \iuatoms)}) = -(\up{(-(\down{R}) \cap \iuatoms)})$,
  \item $\up{\fus(R)} = \dual{\up{\fis(R)}} = -(\down{(\icpl{\fis(R)})}) = -(\down{(\icpl{(\down{R})} \cap \iiatoms)})$,
  \item $\fus(R) = -(\up{(-(\down{R}) \cap \iuatoms)}) \cap -(\down{(\icpl{(\down{R})} \cap \iiatoms)})$.
  \end{enumerate}
\end{lemma}

\begin{proof}
  For (1), $\down{R} \cap \iuatoms = R \Omega^\converse \cap U 1 = R (\Omega^\converse \cap U 1)$, hence it suffices to show
  \begin{equation*}
    {\ni} 1 = \Omega^\converse \cap U 1 = ({\ni} / {\ni}) \cap U 1.
  \end{equation*}
  The inclusion $\subseteq$ follows by residuation from ${\ni} 1 {\ni} = \alpha(\eta({\ni})) = {\ni}$.
  The opposite inclusion follows from $-{\ni} 1 \subseteq -{\ni} {\in} = -({\ni}/{\ni})$ using boolean properties.

  For (2),
  \begin{align*}
        -(\up{(-\fis(R) \cap \iuatoms)})
    & = -(-(R {\ni} 1) 1^\converse 1 \Omega) \\
    & = -(-(R {\ni} 1 1^\converse) {\in}) \\
    & = -(-(R {\ni}) {\in}) \\
    & = R {\ni} / {\ni} \\
    & = \Lambda(R {\ni}) {\ni} / {\ni} \\
    & = \Lambda(R {\ni}) ({\ni} / {\ni}) \\
    & = \down{\fus(R)}
  \end{align*}
  using Lemma~\ref{lemma.det-fix-ref} in the first step.
  The second equality follows by (1) and boolean algebra.

  For (3),
  \begin{align*}
        \dual{\up{\fis(R)}}
    & = -\icpl{(\up{\fis(R)})} \\
    & = -(\down{(\icpl{\fis(R)})}) \\
    & = -(\fis(R) C \Omega^\converse) \\
    & = -(\alpha(R) 1 \Omega C) \\
    & = -(\alpha(R) {\in} C) \\
    & = -(\alpha(R) (-{\in})) \\
    & = \alpha(R)^\converse \backslash {\in} \\
    & = \Lambda(\alpha(R)) \Omega \\
    & = \up{\fus(R)}
  \end{align*}
  using Lemma~\ref{lemma.rel-mod-props}.
  The remaining equality follows again by (1).

  Item (4) follows from (2) and (3) since $\fus(R)$ is convex-closed.
\end{proof}

\subsection{Properties of approximation and determinisation}
\label{section.fus-fis-props}

Lemma~\ref{lemma.det-prop}(3) generalises to arbitrary multirelations.

\begin{lemma}
  \label{lemma.alpha-props}
  Let $R : X \rto \Pow Y$ and $S : Y \rto \Pow Z$.
  Then
  \begin{enumerate}
  \item $\alpha(R_\seq) = \alpha(\dom(R)_\seq) \alpha(R)$,
  \item $\alpha(R \seq S) \subseteq \alpha(R) \alpha(S)$,
  \item $\alpha(\down{R}) = \alpha(R)$.
  \end{enumerate}
\end{lemma}

\begin{proof}
  For (1), ${\ni} \alpha(T) = \alpha(\Lambda({\ni} T {\ni})) = \alpha(T_\kleisli)$.
  Thus
  \begin{align*}
        \alpha(R_\seq)
    & = \dom(R)_\seq \bigcup_{T \subseteq_d R} \alpha(T_\kleisli) \\
    & = \dom(R)_\seq \bigcup_{T \subseteq_d R} {\ni} \alpha(T) \\
    & = \dom(R)_\seq {\ni} \alpha(R) \\
    & = \alpha(\dom(R)_\seq) \alpha(R).
  \end{align*}

  This implies (2) by $\alpha(R \seq S) = R \alpha(S_\seq) \subseteq R \alpha(1) \alpha(S) = \alpha(R) \alpha(S)$.

  For (3), $\alpha(\down{R}) = R \Omega^\converse {\ni} = R {\ni} = \alpha(R)$ since ${\in} \Omega = {\in}$.
\end{proof}

Part (2) of the previous lemma cannot in general be strengthened to an equality.
Of course, structure is lost when approximating.

\begin{example}
  \label{example.alpha-counter}
  For $R = \{ (a,\{a,b\}) \}$,
  \begin{equation*}
      \alpha(R \seq R)
    = \alpha(\emptyset)
    = \emptyset
    \subset \{ (a,a), (a,b) \}
    = \alpha(R)
    = \alpha(R) \alpha(R).
  \end{equation*}
\end{example}

\begin{lemma}
  \label{lemma.fusion-fission-props}
  $\fis \circ \fis = \fis$, $\fus \circ \fus = \fus$, $\fis \circ \fus = \fis$ and $\fus \circ \fis = \fus$.
\end{lemma}

\begin{proof}
  The first two properties are part of the closure conditions in Proposition~\ref{proposition.eta-alpha-lambda-galois}.
  The proof of the remaining ones are similar.
\end{proof}

\begin{lemma}
  \label{lemma.idet-peleg-alpha}
  Let $R : X \rto \Pow Y$ and $S : Y \to \Pow Z$.
  Then
  \begin{enumerate}
  \item $\fis(R) \seq S = \alpha(R) S$,
  \item $\fis(R \seq S) \subseteq \fis(R) \seq \fis(S)$,
  \item $\fus(R \seq S) \subem \fus(R) \seq \fus(S)$.
  \end{enumerate}
 \end{lemma}

\begin{proof}
  For (1), $\fis(R) \seq S = \alpha(R) 1 S_\seq = \alpha(R) (1 \seq S) = \alpha(R) S$.

  For (2), $\fis(R \seq S) = \eta(\alpha(R \seq S)) \subseteq \eta(\alpha(R) \alpha(S)) = \eta(\alpha(R)) \seq \eta(\alpha(S)) = \fis(R) \seq \fis(S)$ using Lemmas~\ref{lemma.alpha-props} and~\ref{lemma.det-prop}.

  For (3), $\fus(R \seq S) = \Lambda(\alpha(R \seq S)) \subem \Lambda(\alpha(R) \alpha(S)) = \Lambda(\alpha(R)) \seq \Lambda(\alpha(S)) = \fus(R) \seq \fus(S)$ using Lemmas~\ref{lemma.alpha-props},~\ref{lemma.lambda-em-isotone} and~\ref{lemma.det-prop}.
\end{proof}

The proofs of the following lemma are immediate from properties of Section~\ref{section.relation-multirelation}.

\begin{lemma}
  \label{lemma.fusion-kleisli}
  Let $R : X \rto \Pow Y$ and $S : X \rto Y$.
  Then
  \begin{enumerate}
  \item $R_\kleisli = \fus({\ni} R)$,
  \item $\fus(R) = \eta R_\kleisli = \Lambda(R) \mu$.
  \end{enumerate}
\end{lemma}

The following diagram collects a number of identities for $\fus$, $\Pow$, $(-)_\Pow$ and $\mu$ from our Isabelle theories, most of which have been mentioned in previous sections.

\begin{equation*}
  \begin{tikzcd}
    X \rto \Pow^2 Y \ar[rr, "(-)_\kleisli"] \ar[d, "(-) \mu_Y"'] & & \Pow X \rto \Pow^2 Y \ar[d, "(-) \mu_Y"] \\[4ex]
    X \rto \Pow Y \ar[d, "\alpha"'] & X \rto \Pow Y \ar[l, "\fus"] \ar[ul, "\Lambda"'] \ar[ur, "\Pow"] \ar[r, "(-)_\kleisli"'] \ar[d, "\alpha"'] & \Pow X \rto \Pow Y \ar[ll, bend right=15, crossing over, "1 (-)"'] \\[4ex]
    X \rto Y & X \rto Y \ar[l, "\id"] \ar[ul, "\Lambda"] \ar[ur, "\Pow"'] \ar[r, "{\ni} (-)"'] & \Pow X \rto Y \ar[u, "\Lambda"']
  \end{tikzcd}
\end{equation*}

\begin{remark}
  The bottom half of this diagram can be extended similarly by decomposing $\fis(R) = 1_X ({\ni} R {\ni} 1_Y)$, since $1_X {\ni} = \Id_X$.
  The intermediate relation ${\ni} R {\ni} 1_Y$ corresponds to the Kleisli lifting of $R$ except that it maps to all singleton subsets instead of the union.

  Like $\fus$, the operators $\Pow$, $(-)_\Pow$ and the constant $\mu$ all factor through $\Lambda$, that is, they first produce a set of sets which is then flattened by taking their union resulting in outer deterministic multirelations.
  This dualises to $\fis$ using corresponding inner operators $\Pow_i(R) = \eta({\ni} R)$, $R_{\Pow_i} = \Pow_i(\alpha(R))$ and constant $\mu_i = \Pow_i({\ni})$ which factor through $\eta$, that is, flatten the set of sets to singleton sets resulting in inner deterministic multirelations.
  The entire extended diagram dualises this way.
  We leave further study of these inner operators and their structural role to future work.
\end{remark}


\section{Category of Inner Univalent Multirelations}
\label{section.unival-cat}

It remains to describe the category of inner univalent multirelations.
This requires different techniques, as this category is not isomorphic to $\Rel$.

Inner total multirelations have previously been called \emph{non-terminal} multirelations~\cite{FurusawaStruth2016}, writing $\nu(R)$ for the set of \emph{non-terminal} elements of $R$: those pairs in $R$ whose second component is not $\emptyset$, that is, $\nu(R) = R - \iuone$.
In addition, the map $\tau(R) = R \seq \emptyset = R \cap \iuone$ projects on the \emph{terminal} elements of $R$: those pairs in $R$ whose second component is $\emptyset$.
We use this notation in this section.

First we list properties needed for reasoning about $\alpha$, $\fis$ and $\fus$ in the presence of the non-terminal and terminal multirelations.

\begin{lemma}
  \label{lemma.nu-tau-props}
  Let $R : X \rto \Pow Y$.
  Then
  \begin{enumerate}
  \item $\alpha(\tau(R)) = \emptyset$ and $\alpha(\nu(R)) = \alpha(R)$,
  \item $\nu(\fis(R)) = \fis(R) = \fis(\nu(R))$ and $\tau(\fis(R)) = \emptyset$,
  \item $\fus(\nu(R)) = \fus(R)$,
  \item $R \seq S = \nu(R) \seq S \cup \tau(R)$,
  \item $\tau(R \seq S) = \tau(R) \cup \nu(R) \seq \tau(S)$.
  \end{enumerate}
\end{lemma}

\begin{proof}
  For (1), $\alpha(\tau(R)) = R \emptyset_\seq {\ni} = R \emptyset{\ni} = \emptyset$.
  The second property follows from $R = \nu(R) \cup \tau(R)$.
  Both properties in (2) are obvious, and so is (3).
  Items (4) and (5) are from~\cite{FurusawaStruth2016}.
\end{proof}

\begin{remark}
  The multirelation $R = \{ (a,\emptyset) \}$ on $X = \{a,b\}$ satisfies
  \begin{equation*}
    \nu(\fus(R)) = \emptyset \neq R \cup \{ (b,\emptyset) \} = \fus(R)
  \end{equation*}
  because $\fus$ adds a pair $(c,\emptyset)$ for each $c$ that is not related to any set.
\end{remark}

\begin{lemma}
  \label{lemma.iuni-props}
  Let $R : X \rto \Pow Y$.
  Then the following statements are equivalent:
  \begin{enumerate}
  \item $R$ is inner univalent,
  \item $\nu(R) \subseteq \iuatoms$,
  \item $\nu(R) = \fis(R)$,
  \item $R = \fis(R) \cup \tau(R)$.
  \end{enumerate}
\end{lemma}

\begin{proof}
  Suppose (1), that is, $R \subseteq \iuone \cup \iuatoms$.
  Then $\nu(R) = R - \iuone \subseteq \iuatoms$ yields (2) by boolean algebra, and further $\nu(R) \subseteq R \cap \iuatoms \subseteq \down{R} \cap \iuatoms = \fis(R)$.
  For (3) it then remains to show $\fis(R) \subseteq \nu(R)$ and thus $\down{R} \cap \iuatoms \subseteq R$ by Lemma~\ref{lemma.fus-fis-explicit}.
  Further, by Lemmas~\ref{lemma.det-fix-ref} and because down-closure of multirelations preserves unions~\cite[Lemma 4.3]{FurusawaGuttmannStruth2023a},
  \begin{equation*}
    \down{R} = \down{(R \cap (\iuatoms \cup \iuone))} = \down{(R \cap \iuatoms)} \cup \down{(R \cap \iuone)}
  \end{equation*}
  Since $\down{(R \cap \iuone)} \cap \iuatoms \subseteq \down{\iuone} \cap \iuatoms = \iuone \cap \iuatoms = \emptyset$, it remains to show $\down{(R \cap \iuatoms)} \cap \iuatoms \subseteq R$.
  Using $\iuatoms = U 1$, we have
  \begin{equation*}
    \down{(R \cap \iuatoms)} \cap \iuatoms = (R \cap U 1) \Omega^\converse \cap U 1 = R (\Omega^\converse \cap 1^\converse U \cap U 1) = R (\Omega^\converse \cap 1^\converse U 1)
  \end{equation*}
  so it suffices to show $\Omega^\converse \cap 1^\converse U 1 \subseteq \Id$.
  This is obtained by
  \begin{equation*}
    \Omega^\converse \cap 1^\converse U 1 \subseteq 1^\converse (U 1 \cap 1 \Omega^\converse) = 1^\converse 1 (\Omega^\converse \cap U 1) \subseteq 1^\converse 1 \Omega^\converse 1^\converse 1 = 1^\converse 1 \subseteq \Id
  \end{equation*}
  if we can show $1 \Omega^\converse 1^\converse = \Id$.
  The latter follows by applying converse to
  \begin{equation*}
    1 \Omega 1^\converse = 1 ({\in} \backslash {\in}) 1^\converse = ({\in} 1^\converse \backslash {\in}) 1^\converse = {\in} 1^\converse \backslash {\in} 1^\converse = \Id \backslash \Id = \Id
  \end{equation*}
  From (3) we immediately obtain (4) by partitioning $R$ into terminal and non-terminal elements.
  Finally, (4) implies (1) because $R = \fis(R) \cup \tau(R) \subseteq \iuatoms \cup \iuone$ using Lemma~\ref{lemma.fus-fis-explicit}.
\end{proof}

It follows that $\fis(R) \subseteq R$ for inner univalent $R$.

\begin{lemma}
  \label{lemma.peleg-alpha-tau}
  Let $R$ be inner univalent.
  Then
  \begin{enumerate}
  \item $R \seq S = \alpha(R) S \cup \tau(R) $,
  \item $\alpha(R \seq S) = \alpha(R) \alpha(S)$,
  \item $\fis(R \seq S) = \fis(R) \seq \fis(S)$ and $\fus(R \seq S) = \fus(R) \seq \fus(S)$.
  \end{enumerate}
\end{lemma}

\begin{proof}
  Suppose $R$ is inner univalent.
  For (1), therefore, $\nu(R) = \fis(R)$ by Lemma~\ref{lemma.iuni-props}.
  Hence, with Lemmas~\ref{lemma.idet-peleg-alpha},~\ref{lemma.nu-tau-props} and~\ref{lemma.iuni-props}, $R \seq S = \nu(R) \seq S \cup \tau(R) = \fis(R) \seq S \cup \tau(R) = \alpha(R) S \cup \tau(R)$.

  For (2), $\alpha(R \seq S) = \alpha(\alpha(R) S \cup \tau(R)) = \alpha(\alpha(R) S) \cup \alpha(\tau(R))= \alpha(R) \alpha(S)$ using (1) and Lemma~\ref{lemma.nu-tau-props}.

  Item (3) is then immediate from (2) and Lemma~\ref{lemma.det-prop}.
\end{proof}

Items (2) and (3) generalise parts of Lemma~\ref{lemma.det-prop} and Corollary~\ref{cor.det-cat2} from inner determinism to inner univalence.

\begin{proposition}
  \label{proposition.peleg-rel-inner-univalent}
  The inner univalent multirelations form a category with composition $\seq$ and identity arrows $1_X$.
\end{proposition}

\begin{proof}
  Suppose $Q$, $R$ and $S$ are composable and inner univalent.
  For closure under Peleg composition, the $1_X$ are clearly inner univalent and $R \seq S$ is inner univalent because
  \begin{align*}
        R \seq S
    & = \alpha(R) (\fis(S) \cup \tau(S)) \cup \tau(R) \\
    & = \fis(R) \seq \fis(S) \cup \fis(R) \seq \tau(S) \cup \tau(R) \\
    & = \fis(R \seq S) \cup \tau(R \seq S),
  \end{align*}
  using Lemmas~\ref{lemma.idet-peleg-alpha},~\ref{lemma.nu-tau-props},~\ref{lemma.iuni-props} and~\ref{lemma.peleg-alpha-tau}.
  For associativity, using the same lemmas,
  \begin{align*}
        Q \seq (R \seq S)
    & = \alpha(Q) (\alpha(R) S \cup \tau(R)) \cup \tau(Q) \\
    & = \alpha(Q) \alpha(R) S \cup \alpha(Q) \tau(R) \cup \tau(Q) \\
    & = \alpha(Q \seq R) S \cup \tau (Q \seq R) \\
    & = (Q \seq R) \seq S.
    \qedhere
  \end{align*}
\end{proof}

We have thus characterised the categories of inner and outer deterministic as well as those of inner and outer univalent multirelations within the algebra of multirelations.

\begin{lemma}
  \label{lemma.quantaloid-inner-univalent}
  Each homset in the category of inner univalent multirelations is a complete lattice.
  In this category, Peleg composition preserves arbitrary sups in the first argument and non-empty sups in the second argument.
\end{lemma}

\begin{proof}
  Inner univalent multirelations are closed under arbitrary unions by Lemma~\ref{lemma.iuni-props} since $\nu$ preserves arbitrary unions.
  This forms a complete lattice structure on the homsets.
  Peleg composition preserves arbitrary unions in its first argument.
  Preservation in the second argument holds for non-empty unions:
  \begin{equation*}
    R \seq \bigcup_{i \in I} S_i = \left( \alpha(R) \bigcup_{i \in I} S_i \right) \cup \tau(R) = \bigcup_{i \in I} (\alpha(R) S_i \cup \tau(R)) = \bigcup_{i \in I} R \seq S_i
  \end{equation*}
  by Lemma~\ref{lemma.peleg-alpha-tau} if $I \neq \emptyset$.
\end{proof}

\begin{remark}
  Inner univalent multirelations do not form quantaloids with respect to sups, that is, Proposition~\ref{proposition.peleg-rel} does not generalise beyond Lemma~\ref{lemma.quantaloid-inner-univalent}.
  For $I = \emptyset$ we have $R \seq \bigcup_{i \in I} S_i = R \seq \emptyset$ and $\bigcup_{i \in I} (R \seq S_i) = \emptyset$.
  But $R \seq \emptyset = \emptyset$ if and only if $\tau(R) = \emptyset$, that is, $R$ must be inner total.
  This shows that $R$ must be inner deterministic for this argument to work.

  Likewise, outer univalent multirelations do not form quantaloids with respect to $\iU$.
  A counterexample again uses $I = \emptyset$ in which case $\emptyset \seq \iU_{i \in I} S_i = \emptyset$ but $\iU_{i \in I} (\emptyset \seq S_i) = \iuone$.
\end{remark}

\begin{remark}
  Lemma~\ref{lemma.peleg-alpha-tau} and Proposition~\ref{proposition.peleg-rel-inner-univalent} show that $\alpha$, $\fis$ and $\fus$ are functors from the category of inner univalent multirelations to $\Rel$ and the categories of inner and outer deterministic multirelations, respectively.
  They need not be injective.
  For $X = \{a,b\}$, for instance, $\alpha$ maps the inner univalent multirelations $\{ (a,\emptyset) \}$ and $\{ (a,\emptyset), (b,\emptyset) \}$ in $X \to \Pow Y$ to the relation $\emptyset_{X,Y}$.
  This failure of injectivity extends along $\Lambda$ and $\eta$, so that the categories are not isomorphic.
\end{remark}


\section{A Fine-Grained View on Determinisation}
\label{section.fine-grained}

Many properties of inner deterministic multirelations hold already of inner univalent ones.
Here we prove refined results.
First we refine Corollary~\ref{cor.det-fix2} for deterministic multirelations.

\begin{lemma}
  \label{lemma.det-fix}
  ~
  \begin{enumerate}
  \item The outer univalent multirelations are precisely the postfixpoints of $\fus$ with respect to $\subseteq$ and the prefixpoints of $\fus$ with respect to $\subs$.
  \item Prefixpoints of $\fus$ with respect to $\subseteq$ and $\subh$ and postfixpoints with respect to $\subs$ are outer total.
        The postfixpoints of $\fus$ with respect to $\subem$ and $\subs$ coincide.
  \item Prefixpoints of $\fus$ with respect to $\subem$ are outer deterministic.
  \end{enumerate}
\end{lemma}

\begin{proof}
  For (1), if $R$ is univalent, then
  \begin{equation*}
    R = R (\syq{\in}{\in}) = (\syq{{\in} R^\converse}{\in}) \cap R U \subseteq \syq{{\in} R^\converse}{\in} = \Lambda(\alpha(R)) = \fus(R).
  \end{equation*}
  Conversely, if $R$ is a postfixpoint of $\fus$ with respect to $\subseteq$, then
  \begin{equation*}
    R^\converse R \subseteq \fus(R)^\converse \fus(R) = \Lambda(\alpha(R))^\converse \Lambda(\alpha(R)) = (\syq{\in}{{\in} R^\converse}) (\syq{{\in} R^\converse}{\in}) \subseteq \syq{\in}{\in} = \Id.
  \end{equation*}
  Since $\fus(R) \subseteq \up{\fus(R)}$, postfixpoints of $\fus$ with respect to $\subseteq$ are also prefixpoints with respect to $\subs$.
  Conversely, if $R$ is a prefixpoint with respect to $\subs$, then
  \begin{equation*}
    R \subseteq \up{\fus(R)} = \Lambda(\alpha(R)) \Omega = \Lambda(\alpha(R)) ({\in} \backslash {\in}) = {\in} \Lambda(\alpha(R))^\converse \backslash {\in} = {\in} (\syq{\in}{{\in} R^\converse}) \backslash {\in} = {\in} R^\converse \backslash {\in}.
  \end{equation*}
  Together with $R \subseteq \alpha(R) / {\ni}$, we obtain $R \subseteq \syq{{\in} R^\converse}{\in} = \Lambda(\alpha(R)) = \fus(R)$.
  Thus $R$ is a postfixpoint with respect to $\subseteq$.

  For (2), if $R$ is a postfixpoint of $\fus$ with respect to $\subs$, then $\fus(R) \subseteq \up{R}$.
  Hence totality of $R$ follows by $U = \Lambda(\alpha(R)) U = \fus(R) U \subseteq \up{R} U = R \Omega U = R U$ using that $\Lambda$ yields deterministic multirelations and $\Omega$ is total.
  The proof for prefixpoints with respect to $\subh$ is similar, using $\Omega^\converse$ instead of $\Omega$.
  Moreover prefixpoints with respect to $\subseteq$ are also postfixpoints with respect to $\subs$ since $R \subseteq \up{R}$.
  The remaining claim follows since any $R$ is a postfixpoint of $\fus$ with respect to $\subh$ (Proposition~\ref{proposition.eta-alpha-lambda-galois}).

  Finally, (3) follows by (1) and (2).
\end{proof}

Obviously, if $R$ is outer deterministic, then $\fus(R) = \Lambda(\alpha(R)) = R$.
The converse implication follows by (3) above.
This yields an alternative algebraic proof of the fact that the outer deterministic multirelations are precisely the fixpoints of $\fus$.

Next we refine Corollary~\ref{cor.det-fix2} for inner deterministic multirelations.

\begin{lemma}
  \label{lemma.idet-fix}
  ~
  \begin{enumerate}
  \item Inner univalent multirelations are prefixpoints of $\fis$ with respect to $\subseteq$ and postfixpoints of $\fis$ with respect to $\subs$.
  \item Postfixpoints of $\fis$ with respect to $\subh$ are inner univalent.
        The postfixpoints of $\fis$ with respect to $\subem$ and $\subh$ coincide.
  \item The inner total multirelations are precisely the prefixpoints of $\fis$ with respect to $\subs$.
        The prefixpoints of $\fis$ with respect to $\subem$ and $\subs$ coincide.
  \item Postfixpoints of $\fis$ with respect to $\subseteq$ are inner deterministic.
  \end{enumerate}
\end{lemma}

\begin{proof}
  ~
  (1) follows by Lemma~\ref{lemma.iuni-props}, using $\fis(R) \subseteq R \subseteq \up{R}$.

  For (2), if $R$ is a postfixpoint of $\fis$ with respect to $\subh$, then $R \subseteq \down{\fis(R)} \subseteq \down{\iuatoms} = \iuone \cup \iuatoms$, so $R$ is inner univalent.
  Hence by (1), $R$ is also a postfixpoint with respect to $\subs$ and therefore with respect to $\subem$.

  For (3), if $R$ is a prefixpoint of $\fis$ with respect to $\subs$, then $R \subseteq \up{\fis(R)} \subseteq \up{\iuatoms} = -\iuone$, so $R$ is inner total.
  Conversely, if $R$ is inner total, then
  \begin{equation*}
    R \subseteq R R^\converse R \subseteq R (-\iuone)^\converse (-\iuone) = R {\ni} {\in} = R {\ni} 1 \Omega = \fis(R) \Omega = \up{\fis(R)}.
  \end{equation*}
  The claim follows since any $R$ is a prefixpoint of $\fis$ with respect to $\subh$ by Proposition~\ref{proposition.eta-alpha-lambda-galois}.

  For (4), since $\fis(R) \subseteq \down{\fis(R)}$ and $\fis(R) \subseteq \up{\fis(R)}$, postfixpoints of $\fis$ with respect to $\subseteq$ are also postfixpoints with respect to $\subh$ and prefixpoints with respect to $\subs$.
  Hence the final claim follows by (2) and (3).
\end{proof}

To show that the inner deterministic multirelations are precisely the fixpoints of $\fis$, it remains to check, using parts (1) and (4) above, that inner deterministic multirelations are postfixpoints of $\fis$ with respect to $\subseteq$.
Indeed, if $R$ is inner deterministic, then $R \subseteq \iuatoms = U 1$.
Hence
\begin{equation*}
  R = R \cap U 1 \subseteq R 1^\converse 1 \subseteq R {\ni} 1 = \fis(R)
\end{equation*}
using $1 \subseteq {\in}$.

The following results revisit previous closure results in the context of total multirelations.

\begin{lemma}
  \label{lemma.peleg-preserves-inner-univalent}
  Inner and outer total multirelations are closed under Peleg composition.
\end{lemma}

\begin{proof}
  Let $R$ and $S$ be total.
  Then
  \begin{equation*}
    S_\seq = \dom(S)_\seq \bigcup_{Q \subseteq_d R} Q_\kleisli = 1_\seq \bigcup_{Q \subseteq_d R} Q_\kleisli = \bigcup_{Q \subseteq_d R} (\syq{{\in} Q {\in}}{\in}).
  \end{equation*}
  Hence
  \begin{equation*}
    S_\seq S_\seq^\converse = \bigcup_{P, Q \subseteq_d R} (\syq{{\in} P {\in}}{\in}) (\syq{\in}{{\in} Q {\in}}) = \bigcup_{P, Q \subseteq_d R} (\syq{{\in} P {\in}}{{\in} Q {\in}}) \supseteq \bigcup_{P \subseteq_d R} (\syq{{\in} P {\in}}{{\in} P {\in}}) \supseteq \Id.
  \end{equation*}
  Thus $(R \seq S) (R \seq S)^\converse = R S_\seq S_\seq^\converse R^\converse \supseteq R R^\converse \supseteq \Id$.

  Let $R$ and $S$ be inner total, that is, $R \subseteq -\iuone$ and $S \subseteq -\iuone$.
  Since Peleg composition preserves $\subseteq$ it suffices to show $-\iuone \seq -\iuone \subseteq -\iuone$.
  We have
  \begin{equation*}
    -\iuone \seq -\iuone = -\iuone (-\iuone)_\seq = -\iuone \dom(-\iuone)_\seq \bigcup_{Q \subseteq_d -\iuone} Q_\kleisli = -\iuone \bigcup_{Q \subseteq_d -\iuone} Q_\kleisli = \bigcup_{Q \subseteq_d -\iuone} -\iuone Q_\kleisli.
  \end{equation*}
  Hence it remains to show $-\iuone Q_\kleisli \subseteq -\iuone$ for any $Q \subseteq_d -\iuone$.
  The latter condition means $Q$ is univalent, total and inner total.
  The remaining goal is equivalent to $\iuone (Q_\kleisli)^\converse \subseteq \iuone$.
  This follows by
  \begin{align*}
        \iuone (Q_\kleisli)^\converse
    & = \iuone \Lambda({\ni} Q {\ni})^\converse \\
    & = \iuone (\syq{\in}{{\in} Q^\converse {\in}}) \\
    & = (\syq{{\in} \iuone^\converse}{{\in} Q^\converse {\in}}) \\
    & = (\syq{\emptyset}{{\in} Q^\converse {\in}}) \\
    & = \emptyset / {{\ni} Q {\ni}} \\
    & = -(U {\in} Q^\converse {\in}) \\
    & = -(U (-\iuone Q^\converse {\in})) \\
    & \subseteq -(U Q Q^\converse {\in}) \\
    & \subseteq -(U {\in}) \\
    & = \iuone
  \end{align*}
  using that $Q$ is inner total and total.
\end{proof}

Hence closure of inner deterministic multirelations under Peleg composition (Proposition~\ref{prop.det-cat}) also follows by combining Proposition~\ref{proposition.peleg-rel-inner-univalent} and Lemma~\ref{lemma.peleg-preserves-inner-univalent}.
However, composition of total or inner total multirelations need not be associative.
\begin{example}
  Let $X = \{a,b,c\}$ and $R, S : X \rto \Pow X$ with $R = \{ (a,\{a,b\}), (b,\{a\}), (c,\{c\}) \}$ and $S = \{ (a,\{a,b\}), (b,\{a,c\}), (c,\{c\}) \}$.
  Then $(a,\{a,b,c\}) \in R \seq (R \seq S) - (R \seq R) \seq S$.
  Thus outer and inner total multirelations do not form categories.
\end{example}

We conclude with preservation properties for outer total multirelations.

\begin{lemma}
  \label{lemma.alpha-det-prop2}
  Let $R : X \rto \Pow Y$ and $S : Y \rto \Pow Z$ be outer total.
  Then
  \begin{enumerate}
  \item $\alpha(R \seq S) = \alpha(R) \alpha(S)$,
  \item $\fis(R \seq S) = \fis(R) \seq \fis(S)$ and $\fus(R \seq S) = \fus(R) \seq \fus(S)$.
  \end{enumerate}
\end{lemma}

\begin{proof}
  For (1), if $S$ is outer total, then $\dom(S) = \Id$ and $\alpha(S_\seq) = \alpha(1) \alpha(S)$ by Lemma~\ref{lemma.alpha-props}(1).
  The inclusion step in the proof of Lemma~\ref{lemma.alpha-props}(2) then becomes an equality, which shows the claim.
  Item (2) is then immediate from (1) and Lemma~\ref{lemma.det-prop}.
\end{proof}

\begin{example}
  Part (1) of Lemma~\ref{lemma.alpha-det-prop2} does not hold for outer univalent multirelations: we have $\alpha(1 \seq \emptyset) = \alpha(1) = \Id$ but $\alpha(1) \alpha(\emptyset) = \Id \, \emptyset = \emptyset$.
  For (2) consider $X = \{a,b\}$ and $R, S : X \rto \Pow X$ with $R = \{ (a,\{a,b\}) \}$ and $S = \{ (a,\{a\}) \}$.
  Then $\fis(R \seq S) = \fis(\emptyset) = \emptyset$ but
  \begin{equation*}
    \fis(R) \seq \fis(S) = \{ (a,\{a\}), (a,\{b\}) \} \seq S = \{ (a,\{a\}) \}.
  \end{equation*}
  Moreover
  $\fus(R \seq S) = \fus(\emptyset) = \{ (a,\emptyset), (b,\emptyset) \}$ but
  \begin{equation*}
    \fus(R) \seq \fus(S) = (R \cup \{ (b,\emptyset) \}) \seq (S \cup \{ (b,\emptyset) \}) = \{ (a,\{a\}), (b,\emptyset) \}.
  \end{equation*}
\end{example}


\section{Co-Determinisation and Further Isomorphisms}
\label{subsection.cofusion}

We discuss two additional operations called \emph{co-fusion} and \emph{co-fission} of multirelations.
They are obtained via the inner isomorphism:
\begin{equation*}
  \cofusion{R} = \icpl{\fus(\icpl{R})} \qquad
  \text{and} \qquad
  \cofission{R} = \icpl{\fis(\icpl{R})}.
\end{equation*}
It follows that
\begin{gather*}
  \cofusion{R} = \left\{ (a,B) \mid B = \bigcap R(a) \right\} \qquad
  \text{and} \qquad
  \cofission{R} = \up{R} \cap \iiatoms.
\end{gather*}
Other properties of co-fusion and co-fission follow immediately from the inner isomorphism.
For instance,
\begin{equation*}
  \cofusion{R} = -\down{(-\up{R} \cap \iiatoms)} \cap -\up{(\icpl{\up{R}} \cap \iuatoms)} = -\down{(-\up{R} \cap \iiatoms)} \cap -\up{((\up{R} \cap \iiatoms) \seqint \icpl{1})}
\end{equation*}
and there is a Galois connection defined with respect to $\subs$.
See our Isabelle theories for details.

In Section~\ref{subsection.fusion} we have used isomorphisms to represent $\alpha(R)$ as (inner) deterministic multirelations $\fus(R)$ and $\fis(R)$.
We now discuss further isomorphic representations.
Since deterministic multirelations are isomorphic to their down-closures, we obtain the down-closed representation $\delta_\downarrow(R) = \down{\fus(R)}$.
Application of fusion takes us back according to $\fus(R) = \fus(\delta_\downarrow(R))$.
Moreover, fission is now obtained by $\fis(R) = \delta_\downarrow(R) \cap \iuatoms$.
Hence the down-closed representation contains both the fusion and the fission.
Finally, $\down{(R \seq S)} = R \seq \down{S} = R \seq (\iuone \cup \down{S}) = \down{R} \seq \down{S}$ for deterministic $S$ and $\down{(-)}$ preserves arbitrary inner unions of deterministic multirelations.
Hence we obtain an isomorphism between the quantaloid of deterministic multirelations with $\iU$ and the quantaloid of down-closures of deterministic multirelations with $\iU$.

Deterministic multirelations are also isomorphic to their up-closures, so we obtain the up-closed representation $\delta_\uparrow(R) = \up{\fus(R)}$.
To get back we apply co-fusion: $\fus(R) = \cofusion{\delta_\uparrow(R)}$.
Furthermore, co-fission is obtained by $\cofission{R} = \delta_\uparrow(R) \cap \iiatoms$.
Hence the up-closed representation contains both the fusion and the co-fission.
By~\cite[Example 4.9]{FurusawaGuttmannStruth2023a}, up-closure does not distribute over Peleg composition for deterministic multirelations, so there is no quantaloid isomorphism in this case.
Nevertheless, since $\up{(-)}$ preserves arbitrary inner unions of deterministic multirelations, we obtain complete lattice isomorphisms between deterministic multirelations and their up-closures.

The range of $\cofusion{-}$ equals that of $\fus$, namely the deterministic multirelations.
Hence similar results are obtained by starting from $\cofusion{R}$ instead of $\fus(R)$ and considering down-/up-closed representations.
From the range of $\fis$ it is also possible to apply $\down{}$ or $\up{}$ and then go back by $\fis$ or intersection with $\iuatoms$.
A similar construction applies to the range of $\cofission{-}$.
Note that $\up{(R \seq S)} = \up{R} \seq \up{S}$ for inner-deterministic $R$~\cite[Lemma 4.8]{FurusawaGuttmannStruth2023a} and $\up{(-)}$ preserves arbitrary unions of inner deterministic multirelations~\cite[Lemma 4.3]{FurusawaGuttmannStruth2023a}.
Hence we obtain an isomorphism between the quantaloid of inner deterministic multirelations with $\bigcup$ and the quantaloid of up-closures of inner deterministic multirelations with $\bigcup$.
Example 4.10 in~\cite{FurusawaGuttmannStruth2023a} rules out a corresponding result for the down-closure of inner deterministic multirelations.
However, since $\down{(-)}$ preserves arbitrary unions of inner deterministic multirelations, at least we obtain complete lattice isomorphisms between inner deterministic multirelations and their down-closures.


\section{Conclusion}

We have studied the inner structure of multirelations through the categories of outer and inner univalent and deterministic multirelations and determinisation maps in a multirelational language that combines features of relation algebra and power allegories with multirelational concepts.
Our results add to previous work on the various inner and outer operations on multirelations~\cite{FurusawaGuttmannStruth2023a}, but shift the focus towards univalence and determinism and from a relation-algebraic language to one based on power allegories.

While this multirelational language of has so far been based on concrete relations and multirelations, an axiomatic extension of the abstract allegorical approach, which equips boolean power allegories with multirelational operations, is its most natural continuation.
This would extend Bird and de Moor's algebra of programming \cite{BirdMoor1997} to alternating nondeterminism.
The characterisation of intuitionistic variants of power allegories based on locally complete allegories is another interesting question, or the consideration of categories of multirelations in arbitrary Grothendieck topoi.

Last but not least, the approximation maps $\alpha$, $\fus$ and $\fis$ are important for defining modal operators on multirelations, which arise in concurrent dynamic logic following Peleg~\cite{Peleg1987} and Nerode and Wijesekera~\cite{NerodeWijesekera1990}.
In fact, an algebraic formalisation of such operators in our multirelational language has been a starting point of this line of work.
This is explored in the third part of this trilogy~\cite{FurusawaGuttmannStruth2023c} and yields the main application of our results so far.

\paragraph{Acknowledgement}
Hitoshi Furusawa and Walter Guttmann thank the Japan Society for the Promotion of Science for supporting part of this research through a JSPS Invitational Fellowship for Research in Japan.
Georg Struth would like to thank Yasuo Kawahara for discussions on the allegorical approach to binary relations and his hospitality during several research visits at Kyushu University.


\bibliographystyle{alpha}
\bibliography{multirel}

\begin{thebibliography}{FKST17}

\bibitem[BdM97]{BirdMoor1997}
R.~Bird and O.~de~Moor.
\newblock {\em Algebra of Programming}.
\newblock Prentice Hall, 1997.

\bibitem[FGS23a]{FurusawaGuttmannStruth2023c}
H.~Furusawa, W.~Guttmann, and G.~Struth.
\newblock Modal algebra of multirelations.
\newblock {\em arXiv}, 2305.11346, 2023.
\newblock \url{https://arxiv.org/abs/2305.11346}.

\bibitem[FGS23b]{FurusawaGuttmannStruth2023a}
H.~Furusawa, W.~Guttmann, and G.~Struth.
\newblock On the inner structure of multirelations.
\newblock {\em arXiv}, 2305.11342, 2023.
\newblock \url{https://arxiv.org/abs/2305.11342}.

\bibitem[FKST17]{FurusawaKawaharaStruthTsumagari2017}
H.~Furusawa, Y.~Kawahara, G.~Struth, and N.~Tsumagari.
\newblock Kleisli, {Parikh} and {Peleg} compositions and liftings for
  multirelations.
\newblock {\em Journal of Logical and Algebraic Methods in Programming},
  90:84--101, 2017.

\bibitem[F{\v{S}}90]{FreydScedrov1990}
P.~J. Freyd and A.~{\v{S}}{\v{c}}edrov.
\newblock {\em Categories, Allegories}, volume~39 of {\em North-Holland
  Mathematical Library}.
\newblock Elsevier Science Publishers, 1990.

\bibitem[FS15]{FurusawaStruth2015a}
H.~Furusawa and G.~Struth.
\newblock Concurrent dynamic algebra.
\newblock {\em ACM Transactions on Computational Logic}, 16(4:30):1--38, 2015.

\bibitem[FS16]{FurusawaStruth2016}
H.~Furusawa and G.~Struth.
\newblock Taming multirelations.
\newblock {\em ACM Transactions on Computational Logic}, 17(4:28):1--34, 2016.

\bibitem[GS23]{GuttmannStruth2023}
W.~Guttmann and G.~Struth.
\newblock Inner structure, determinism and modal algebra of multirelations.
\newblock {\em Archive of Formal Proofs}, 2023.
\newblock Formal proof development,
  \url{https://isa-afp.org/entries/Multirelations_Heterogeneous.html}.

\bibitem[KS18]{KlinSalamanca2018}
B.~Klin and J.~Salamanca.
\newblock Iterated covariant powerset is not a monad.
\newblock {\em Electronic Notes in Theoretical Computer Science}, 341:261--276,
  2018.
\newblock Proceedings of Mathematical Foundations of Programming Semantics
  (MFPS 2018).

\bibitem[NW90]{NerodeWijesekera1990}
A.~Nerode and D.~Wijesekera.
\newblock Constructive concurrent dynamic logic {I}.
\newblock Technical Report Mathematical Sciences Institute 90-43, Cornell
  University, 1990.

\bibitem[Pel87]{Peleg1987}
D.~Peleg.
\newblock Concurrent dynamic logic.
\newblock {\em Journal of the ACM}, 34(2):450--479, 1987.

\bibitem[Pit88]{Pitts1988}
A.~M. Pitts.
\newblock Applications of sup-lattice enriched category theory to sheaf theory.
\newblock {\em Proceedings of the London Mathematical Society}, 57(3):433--480,
  1988.

\bibitem[Rew03]{Rewitzky2003}
I.~Rewitzky.
\newblock Binary multirelations.
\newblock In H.~de~Swart, E.~Or{\l}owska, G.~Schmidt, and M.~Roubens, editors,
  {\em Theory and Applications of Relational Structures as Knowledge
  Instruments}, volume 2929 of {\em Lecture Notes in Computer Science}, pages
  256--271. Springer, 2003.

\bibitem[Ros96]{Rosenthal1996}
K.~I. Rosenthal.
\newblock {\em The Theory of Quantaloids}.
\newblock Addison Wesley Longman Limited, 1996.

\bibitem[Sch11]{Schmidt2011}
G.~Schmidt.
\newblock {\em Relational Mathematics}.
\newblock Cambridge University Press, 2011.

\bibitem[SS89]{SchmidtStroehlein1989}
G.~Schmidt and T.~Str{\"o}hlein.
\newblock {\em Relationen und Graphen}.
\newblock Springer, 1989.

\end{thebibliography}

\appendix

\section{Basis}
\label{section.basis}

As in~\cite{FurusawaGuttmannStruth2023a}, almost every operation in this article can be defined in terms of a basis of 6 operations that mix the relational and the multirelational language: the relational operations $-$, $\cap$, $/$ and the multirelational operations $1$, $\iu$, $\seq$.
Here we extend the list from~\cite{FurusawaGuttmannStruth2023a} with definitions of the operations from power allegories.

\begin{multicols}{3}
\begin{itemize}
\item $R \cup S = -(-R \cap -S)$
\item $R - S = R \cap -S$
\item $\emptyset = R \cap -R$
\item $U = -\emptyset$
\item $\up{R} = R \iu U$
\item ${\in} = \up{1}$
\item $\Id = 1 / 1$
\item $R^\converse = -(-\Id / R)$
\item $S R = -(-S / R^\converse)$
\item ${\ni} = {\in}^\converse$
\item $R \backslash S = (S^\converse / R^\converse)^\converse$
\item $\syq{R}{S} = (R \backslash S) \cap (R^\converse / S^\converse)$
\item $\Lambda(R) = \syq{R^\converse}{\in}$
\item $\Pow(R) = \Lambda({\ni} R)$
\item $R_\kleisli = \Pow(R {\ni})$
\item $\mu = \Id_\kleisli$
\item $\Omega = {\in} \backslash {\in}$
\item $C = \syq{\in}{-\in}$
\item $\icpl{R} = R C$
\item $R \ii S = \icpl{(\icpl{R} \iu \icpl{S})}$
\item $\down{R} = X \ii U$
\item $\convex{R} = \up{R} \cap \down{R}$
\item $\iuone = 1 \ii \icpl{1}$
\item $\iione = \icpl{\iuone}$
\item $\dual{R} = -\icpl{R}$
\item $R \seqint S = \icpl{(R \seq \icpl{S})}$
\item $R_\seq = (\Lambda({\ni} 1) \seq 1^\converse R 1) \mu$
\item $\iuatoms = U 1$
\item $\iiatoms = \icpl{\iuatoms}$
\item $\nu(R) = R - \iuone$
\item $\tau(R) = R \cap \iuone$
\item $\alpha(R) = R {\ni}$
\item $\fis(R) = \down{R} \cap \iuatoms$
\item $\fus(R) = 1 R_\kleisli$
\item $\cofission{R} = \up{R} \cap \iiatoms$
\item $\cofusion{R} = \icpl{\fus(\icpl{R})}$
\item $\dom(R) = \Id \cap R R^\converse$
\item $R \subs S \Leftrightarrow S \subseteq \up{R}$
\item $R \subh S \Leftrightarrow R \subseteq \down{S}$
\item $R \subem S \Leftrightarrow R \subh S \wedge R \subs S$
\end{itemize}
\end{multicols}

We could replace relational intersection $\cap$ with a multirelational intersection variant $\cap$ in the basis: relational $\cap$ is obtained by $R \cap S = \alpha(R 1 \cap S 1)$ which can be defined in terms of multirelational $\cap$ and the rest of the basis.
Yet we do not know whether a multirelational $-$ could replace the relational variant.
See the comments on the list in~\cite{FurusawaGuttmannStruth2023a} for further information.

\end{document}